\documentclass{article}[12pt]
\usepackage{amsmath,amsthm,amssymb,amsfonts,amscd,eucal}

\numberwithin{equation}{section}

\def\bn{{\mathbb N}}

\def\r{\rho}

\def\itm#1{\item{$(#1)$}}

\DeclareMathOperator{\Var}{Var} 
\DeclareMathOperator{\Cov}{Cov}

\def\de{\partial}

\newtheorem{theorem}{Theorem}[section]
\newtheorem{conjecture}{Conjecture}[section]
\newtheorem{proposition}[theorem]{Proposition}
\newtheorem{corollary}[theorem]{Corollary}
\newtheorem{definition}{Definition}[section]
\newtheorem{example}{Example}[section]
\newtheorem{lemma}[theorem]{Lemma}
\theoremstyle{definition}
\newtheorem{remark}{Remark}[section]

\makeindex

\begin{document}

\title{A volume inequality for quantum Fisher information and the
uncertainty principle}

\author{ Paolo Gibilisco\footnote{Dipartimento SEFEMEQ, Facolt\`a di
Economia, Universit\`a di Roma ``Tor Vergata", Via Columbia 2, 00133
Rome, Italy.  Email: gibilisco@volterra.uniroma2.it -- URL:
http://www.economia.uniroma2.it/sefemeq/professori/gibilisco}, Daniele
Imparato\footnote{Dipartimento di Matematica, Politecnico di Torino,
Corso Duca degli Abruzzi 24, 10129 Turin, Italy.  Email:
daniele.imparato@polito.it} \ and Tommaso Isola\footnote{Dipartimento
di Matematica, Universit\`a di Roma ``Tor Vergata", Via della Ricerca
Scientifica, 00133 Rome, Italy.  Email: isola@mat.uniroma2.it -- URL:
http://www.mat.uniroma2.it/$\sim$isola} }

\maketitle

\begin{abstract}

Let $A_1,...,A_N$ be complex self-adjoint matrices and let $\rho$ be a
density matrix.  The Robertson uncertainty principle
$$
{\rm det}\left\{ {\rm Cov}_{\rho}(A_h,A_j) \right\} \geq {\rm det}
\left\{ - \frac{i}{2} {\rm Tr}(\rho [A_h,A_j])\right\}
$$
gives a bound for the quantum generalized covariance in terms of the
commutators $ [A_h,A_j]$.  The right side matrix is antisymmetric and
therefore the bound is trivial (equal to zero) in the odd case
$N=2m+1$.

Let $f$ be an arbitrary normalized symmetric operator monotone
function and let $\langle \cdot, \cdot \rangle_{\rho,f}$ be the
associated quantum Fisher information.  In this paper we conjecture
the inequality
$$
{\rm det}\left\{ {\rm Cov}_{\rho}(A_h,A_j) \right\} \geq {\rm det} \left\{ \frac{f(0)}{2} \langle i[\rho, A_h],i[\rho,A_j] \rangle_{\rho,f}\right\}
$$
that gives a non-trivial bound for any $N \in {\mathbb N}$ using the
commutators $i[\rho, A_h]$.  The inequality has been proved in the
cases $N=1,2$ by the joint efforts of many authors (see the
Introduction).  In this paper we prove the (real) case $N=3$.

\smallskip

\noindent 2000 {\sl Mathematics Subject Classification.} Primary
62B10, 94A17; Secondary 46L30, 46L60.

\noindent {\sl Key words and phrases.} Generalized variance,
uncertainty principle, operator monotone functions, matrix means,
quantum Fisher information.

\end{abstract}


\section{Introduction}

Let $(V, g(\cdot,\cdot))$ be a real inner-product vector space and
suppose that $v_1, ...  , v_N \in V$.  The real $N \times N$ matrix
$G:=\{g(v_h,v_j) \}$ is positive semidefinite and one can define ${\rm
Vol}^g(v_1, ...  , v_N):=\sqrt{{\rm det}\{g(v_h,v_j) \}}$.  If the
inner product depends on a further parameter in such a way that
$g(\cdot,\cdot)=g_{\rho}(\cdot,\cdot)$, we write ${\rm Vol}^g(v_1, ... 
, v_N)={\rm Vol}_{\rho}^g(v_1, ...  , v_N)$.

As an example, consider a probability space $(\Omega, {\cal G},\rho)$
and let $V={\cal L}^2_{\mathbb R}(\Omega, {\cal G},\rho)$ be the space
of square integrable real random variables endowed with the scalar
product given by the covariance ${\rm Cov}_{\rho}(A,B):={\rm
E}_{\rho}(AB)-{\rm E}_{\rho}(A){\rm E}_{\rho}(B)$.  For $A_1, ..., A_N
\in {\cal L}^2_{\mathbb R}(\Omega, {\cal G},\rho)$, $G$ is the well
known covariance matrix and one has
\begin{equation}
{\rm Vol}_{\rho}^{\rm Cov}(A_1, ... , A_N) \geq 0. \label{gvar}
\end{equation}
The expression ${\rm det} \{{\rm Cov}_{\rho}(A_h,A_j\}$ is known as
the generalized variance of the random vector $(A_1,...,A_N)$ and, in
general, one cannot expect a stronger inequality.  For instance, when
$N=1$, $(\ref{gvar})$ just reduces to ${\rm Var}_{\rho}(A) \geq 0$.

In non-commutative probability the situation is quite different due to
the possible non-triviality of the commutators $[A_h,A_j] $.  Let
$M_{n,sa}:=M_{n,sa}(\mathbb{C})$ be the space of all $n \times n$
self-adjoint matrices and let ${\cal D}_n^1$ be the set of strictly
positive density matrices (faithful states).  For $A,B \in M_{n,sa}$
and $\rho \in {\cal D}_n^1$ define the (symmetrized) covariance as
${\rm Cov}_{\rho}(A,B):=1/2[{\rm Tr}(\rho A B)+{\rm Tr}(\rho B
A)]-{\rm Tr}(\rho A)\cdot{\rm Tr}(\rho B)$.  If $A_1, ...  , A_N$ are
self-adjoint matrices one has
\begin{equation}
    {\rm Vol}_{\rho}^{\rm Cov}(A_1, ...  , A_N) \geq \begin{cases} 0,
    & N=2m+1,\\
    {\rm det} \{-\frac{i}{2}{\rm Tr}(\rho [A_h,A_j])\}^{\frac{1}{2}},
    & N=2m.
\end{cases}\label{nup}
\end{equation}
Let us call (\ref{nup}) the ``standard" uncertainty principle to
distinguish it from other inequalities like the ``entropic"
uncertainty principle and similar inequalities.  Inequality
(\ref{nup}) is due to Heisenberg, Kennard, Robertson and Schr\"odinger
for $N=2$ (see \cite{Heisenberg:1927} \cite{Kennard:1927}
\cite{Robertson:1929} \cite{Schroedinger:1930}).  The general case is
due to Robertson (see \cite{Robertson:1934}).  Examples of recent
references where inequality (\ref{nup}) plays a role are given by
\cite{Trifonov:1994} \cite{Trifonov:2000} \cite{Trifonov:2002}
\cite{DodonovDodonovMizrahi:2005} \cite{Daoud:2006}
\cite{JarvisMorgan:2006}.

Suppose one is looking for a general inequality of type (\ref{nup})
giving a bound also in the odd case $N=2m+1$. If one considers the
case $N=1$, it is natural to seek such an inequality in terms of the
commutators $[\rho,A_h]$.

One of the purposes of the present paper is to state a conjecture
regarding an inequality similar to (\ref{nup}) but not trivial for any
$N \in {\mathbb N}$.  Let ${\cal F}_{op}$ be the family of symmetric
normalized operator monotone functions.  To each element $f \in {\cal
F}_{op}$ one may associate a $\rho$-depending scalar product $\langle
\cdot , \cdot \rangle_{\rho,f}$ on the self-adjoint (traceless)
matrices, which is a quantum version of the Fisher information (see
\cite{Petz:1996}).  Let us denote the associated volume by ${\rm
Vol}_{\rho}^f$.  We conjecture that for any $N \in {\mathbb N}^+$
(this is one of the main differences from (1.2)) and for arbitrary
self-adjoint matrices $A_1, ...  , A_N$ one has
\begin{equation}
    {\rm Vol}_{\rho}^{\rm Cov}(A_1, ...  , A_N) \geq
    \left(\frac{f(0)}{2}\right)^{\frac{N}{2}}{\rm
    Vol}_{\rho}^{f}(i[\rho,A_1], ...  , i[\rho,A_N]).  \label{conj}
\end{equation}

The cases $N=1,2$ of inequality (\ref{conj}) have been proved by the
joint efforts of a number of authors in several papers: S. Luo, Q.
Zhang, Z. Zhang (\cite{Luo:2000} \cite{Luo:2003b}
\cite{LuoZZhang:2004} \cite{LuoQZhang:2004} \cite{LuoQZhang:2005}); H.
Kosaki (\cite{Kosaki:2005}); K. Yanagi, S. Furuichi, K. Kuriyama (see
\cite{YanagiFuruichiKuriyama:2005}); F. Hansen (\cite{Hansen:2006b});
P. Gibilisco, D. Imparato, T. Isola (\cite{GibiliscoIsola:2007}
\cite{GibiliscoImparatoIsola:2007a}).

In this paper we discuss the inequality (\ref{conj}) when $N=3$ and
prove it in the real case and for some classes of complex self-adjoint
matrices (including Pauli and generalized Gell-Mann matrices).

It is well known that standard uncertainty principle is a simple
consequence of the Cauchy-Schwartz inequality for $N=2$.  It is worth
to note that for the inequality (\ref{conj}) the same role is played
by the Kubo-Ando inequality
$$
2(A^{-1}+B^{-1})^{-1} \leq m_f(A,B)\leq \frac{1}{2}(A+B)
$$
saying that any operator mean is larger than the harmonic mean and
smaller than the arithmetic mean.

The scheme of the paper is as follows.  In Section \ref{prel} we
describe the preliminary notions of operator monotone functions,
matrix means and quantum Fisher information.  In Section \ref{tilde}
we discuss a correspondence between regular and non-regular operator
monotone functions that is needed in the sequel.  In Section \ref{con}
we state our conjecture, namely the inequality (\ref{conj}); we also
state other two conjectures concerning how the right side depends on
$f \in {\cal F}_{op}$ and the conditions to have equality in
(\ref{conj}).  In Section \ref{L} we discuss the case $N=1$ of
(\ref{conj}) presenting the different available proofs.  In Section
\ref{A} we discuss the case $N=2$; here we prove that, while the
technique employed in \cite{GibiliscoImparatoIsola:2007a} works in
both cases $N=1,2$, the technique used in \cite{Hansen:2006b} does
not.  To this purpose, we show that the generalized variance is not a
concave (neither a convex) function of the state.  In Section
\ref{core} we treat the case $N=3$; we are able to prove the
conjectures for real self-adjoint matrices and in other significant
cases.  In \cite{LuoZZhang:2004} it has been proved that the
Wigner-Yanase metric (correlation) has some advantages on covariance
when one aims to measure entanglement; in Section \ref{ent} we show,
for the sake of completeness, that the above argument holds true for
any regular quantum Fisher information.

\section{Operator monotone functions, matrix means and quantum Fisher
information} \label{prel}

Let $M_n:=M_n(\mathbb{C})$ (resp.  $M_{n,sa}:=M_{n,sa}(\mathbb{C})$)
be the set of all $n \times n$ complex matrices (resp.  all $n \times
n$ self-adjoint matrices).  We shall denote general matrices by
$X,Y,...$ while letters $A,B,...$ will be used for self-adjoint
matrices, endowed with the Hilbert-Schmidt scalar product $\langle A,B
\rangle={\rm Tr}(A^*B)$.  The adjoint of a matrix $X$ is denoted by
$X^{\dag}$ while the adjoint of a superoperator $T:(M_n,\langle
\cdot,\cdot \rangle) \to (M_n ,\langle \cdot,\cdot \rangle)$ is
denoted by $T^*$.  Let ${\cal D}_n$ be the set of strictly positive
elements of $M_n$ and ${\cal D}_n^1 \subset {\cal D}_n$ be the set of
strictly positive density matrices, namely $ {\cal D}_n^1=\{\rho \in
M_n \vert {\rm Tr} \rho=1, \, \rho>0 \} $.  If it is not otherwise
specified, from now on we shall treat the case of faithful states,
namely $\rho>0$.

A function $f:(0,+\infty)\to \mathbb{R}$ is said {\it operator
monotone (increasing)} if, for any $n\in \bn$, and $A$, $B\in M_n$
such that $0\leq A\leq B$, the inequalities $0\leq f(A)\leq f(B)$
hold.  An operator monotone function is said {\it symmetric} if
$f(x)=xf(x^{-1})$ and {\it normalized} if $f(1)=1$.

\begin{definition}
    ${\cal F}_{op}$ is the class of functions $f: (0,+\infty) \to
    (0,+\infty)$ such that

    \itm{i} $f(1)=1$,

    \itm{ii} $tf(t^{-1})=f(t)$,

    \itm{iii} $f$ is operator monotone.
\end{definition}

\begin{example}
    Examples of elements in ${\cal F}_{op}$ are given by the following
    list
    \[
    \begin{array}{rcllrcl}
	f_{RLD}(x)&:=&\frac{2x}{x+1},&&
	f_{WY}(x)&:=&\left(\frac{1+\sqrt{x}}{2}\right)^2,\\[12pt]
	f_{SLD}(x)&:=&\frac{1+x}{2},&& 
	f_{WYD(\beta)}(x)&:=& \beta (1-
	\beta) \frac{(x-1)^2}{(x^{\beta}-1) (x^{1-\beta}-1)},\qquad
	\beta \in \Bigl(0,\frac{1}{2}\Bigr).  
    \end{array}
    \]
\end{example}

 We now report Kubo-Ando theory of matrix means (see
\cite{KuboAndo:1979/80}) as exposed in \cite{PetzTemesi:2005}.

\begin{definition}
 A {\sl mean} for pairs of positive matrices is a function
$m:{\cal D}_n \times {\cal D}_n \to {\cal D}_n$ such that

(i) $m(A,A)=A$,

(ii) $m(A,B)=m(B,A)$,

(iii) $A <B  \quad \Longrightarrow \quad A<m(A,B)<B$,

(vi) $A<A', \quad B<B' \quad \Longrightarrow \quad m(A,B)<m(A',B')$,

(v) $m$ is continuous,

(vi) $Cm(A,B)C^* \leq m(CAC^*,CBC^*)$, for every $ C \in M_n$.
\end{definition}

Property $(vi)$ is known as the transformer inequality. We denote by
$\displaystyle {\cal M}_{op}$ the set of matrix means. The
fundamental result, due to Kubo and Ando, is the following.

\begin{theorem}
    There exists a bijection between ${\cal M}_{op}$ and ${\cal
    F}_{op}$ given by the formula
    $$
    m_f(A,B):= A^{\frac{1}{2}}f(A^{-\frac{1}{2}} B
    A^{-\frac{1}{2}})A^{\frac{1}{2}}.
    $$
\end{theorem}

\begin{example}
    The arithmetic, geometric and harmonic (matrix) means are given
    respectively by
    \[
    \begin{array}{rcl}
	A \nabla B&:=&\frac{1}{2}(A+B),\\[12pt]
	A\# B&:=&A^{\frac{1}{2}}(A^{-\frac{1}{2}} B
	A^{-\frac{1}{2}})^{\frac{1}{2}}A^{\frac{1}{2}}, \\[12pt]
	A{\rm !}B&:=&2(A^{-1}+B^{-1})^{-1}.
    \end{array}
    \]
    They correspond respectively to the operator monotone functions
    $\frac{x+1}{2},\sqrt{x},\frac{2x}{x+1}$.
\end{example}

Kubo and Ando \cite{KuboAndo:1979/80} proved that, among matrix means,
arithmetic is the largest while harmonic is the smallest. 

\begin{proposition}
    For any $f \in {\cal F}_{op}$ one has
    $$
    2(A^{-1}+B^{-1})^{-1} \leq m_f(A,B)\leq \frac{1}{2}(A+B).
    $$
\end{proposition}

\begin{corollary} \label{basic}
    For any $f \in {\cal F}_{op}$ and for any $x,>0$ one has
    $$
    \frac{2x}{1+x}\leq f(x) \leq \frac{1+x}{2}.
    $$
\end{corollary}

In what follows, if ${\cal N}$ is a differential manifold we denote by
$T_{\rho} \cal N$ the tangent space to $\cal N$ at the point $\rho \in
{\cal N}$.  Recall that there exists a natural identification of
$T_{\rho}{\cal D}^1_n$ with the space of self-adjoint traceless
matrices; namely, for any $\rho \in {\cal D}^1_n $
$$
T_{\rho}{\cal D}^1_n =\{A \in M_n|A=A^* \, , \, \hbox{Tr}(A)=0 \}.
$$

 A Markov morphism is a completely positive and trace preserving
 operator $T: M_n \to M_m$.  A {\sl monotone metric} is a family of
 Riemannian metrics $g=\{g^n\}$ on $\{{\cal D}^1_n\}$, $n \in
 \mathbb{N}$, such that
 $$
 g^m_{T(\rho)}(TX,TX) \leq g^n_{\rho}(X,X)
 $$
 holds for every Markov morphism $T:M_n \to M_m$, for every $\rho \in
 {\cal D}^1_n$ and for every $X \in T_\rho {\cal D}^1_n$.  Usually
 monotone metrics are normalized in such a way that $[A,\rho]=0$
 implies $g_{\rho} (A,A)={\rm Tr}({\rho}^{-1}A^2)$.  A monotone metric
 is also said a {\sl quantum Fisher information} (QFI) because of
 Chentsov uniqueness theorem for commutative monotone metrics (see
 \cite{Chentsov:1982}).

 Define $L_{\rho}(A):= \rho A$, and $R_{\rho}(A):= A\rho$, and observe
 that they are commuting self-adjoint (positive) superoperators on
 $M_{n,sa}$.  For any $f\in {\cal F}_{op}$ one can define the positive
 superoperator $m_f(L_{\rho},R_{\rho})$.  Now we can state the
 fundamental theorem about monotone metrics.

\begin{theorem} {\rm \cite{Petz:1996} }

    There exists a bijective correspondence between monotone metrics
    (quantum Fisher informations) on ${\cal D}^1_n$ and normalized
    symmetric operator monotone functions $f\in {\cal F}_{op}$.  This
    correspondence is given by the formula
    $$
   \langle A,B \rangle_{\rho,f}:={\rm Tr}(A\cdot
    m_f(L_{\rho},R_{\rho})^{-1}(B)).
    $$
\end{theorem}

 The metrics associated with the functions $f_{\beta}$ are very
 important in information geometry and are related to
 Wigner-Yanase-Dyson information (see for example
 \cite{GibiliscoIsola:2001} \cite{GibiliscoIsola:2003}
 \cite{GibiliscoIsola:2004} \cite{GibiliscoIsola:2005}
 \cite{GibiliscoImparatoIsola:2007} and references therein).

\section{The function $\tilde f$ and its properties} \label{tilde}

For $f \in {\cal F}_{op}$ define $f(0):=\lim_{x\to 0} f(x)$.  The
condition $f(0)\not=0$ is relevant because it is a necessary and
sufficient condition for the existence of the so-called radial
extension of a monotone metric to pure states (see
\cite{PetzSudar:1996}).  Following \cite{Hansen:2006b} we say that a
function $f \in {\cal F}_{op}$ is {\sl regular} iff $f(0) \not= 0$. 
The corresponding operator mean, associated QFI, etc.  are said
regular too.

\begin{definition}
    We introduce the sets
    $$
    {\cal F}_{op}^{\, r}:=\{f\in {\cal F}_{op}| \quad f(0) \not= 0 \},
    \quad {\cal F}_{op}^{\, n}:=\{f\in {\cal F}_{op}| \quad f(0) = 0
    \}.
    $$
\end{definition}

Trivially one has ${\cal F}_{op}={\cal F}_{op}^{\, r}\dot{\cup}{\cal
F}_{op}^{\, n}$.

\begin{proposition} {\rm \cite{GibiliscoImparatoIsola:2007}}
    For $f \in {\cal F}_{op}^{\, r}$ and $x>0$ set
    $$
    \tilde{f}(x):=\frac{1}{2}\left[ (x+1)-(x-1)^2 \frac{f(0)}{f(x)}
    \right].
    $$
    Then ${\tilde f} \in {\cal F}_{op}^{\, n}$.
\end{proposition}

By the very definition one has the following result.

\begin{proposition} {\rm (\cite{GibiliscoImparatoIsola:2007},
Proposition 5.7)} \label{min}

Let $f \in {\cal F}^r_{op}$.  The following three conditions are
equivalent:

 \itm{1} $\qquad \tilde{f} \leq \tilde{g}$;

 \itm{2} $\qquad m_{\tilde{f}} \leq m_{\tilde{g}}$;

 \itm{3} $\qquad \frac{f(0)}{f(t)} \geq \frac{g(0)}{g(t)} \qquad
 \forall t >0$.
\end{proposition}

Let us give some more definitions.

\begin{definition}
    Suppose that $\rho \in {\cal D}_n^1$ is fixed.  Define
    $X_0:=X-{\rm Tr}(\rho X) I$.
\end{definition}

\begin{definition} 
    For $A,B \in M_{n,sa}$ and $\rho \in {\cal D}_n^1$ define
    covariance and variance as 
    \begin{equation} 
	{\rm Cov}_{\rho}(A,B):=\frac{1}{2}[{\rm Tr}(\rho A B)+{\rm
	Tr}(\rho BA)]-{\rm Tr}(\rho A)\cdot{\rm Tr}(\rho B)=
	\frac{1}{2}[{\rm Tr}(\rho A_0 B_0)+{\rm Tr}(\rho B_0
	A_0)]={\rm Re}\{{\rm Tr}(\rho A_0 B_0)\},\label{cov}
    \end{equation}
    \[ 
    \!{\rm Var}_{\rho}(A):={\rm Cov}_{\rho}(A,A)={\rm Tr}(\rho
    A^2)-{\rm Tr}(\rho A)^2 = {\rm Tr}(\rho A^2_0). 
    \quad\quad\quad\quad\quad\quad\quad\quad\quad\quad
    \quad\quad\quad\quad\quad\quad\quad
    \]
\end{definition}
Suppose, now, that $A,B \in M_{n,sa}$, $\rho \in {\cal D}^1_n$ and
$f \in {\cal F}^r_{op}$.
The fundamental theorem for our present purpose is given by
Proposition 6.3 in \cite{GibiliscoImparatoIsola:2007}, which is stated as
follows.

\begin{theorem} \label{!}
$$
\frac{f(0)}{2} \langle i[\rho,A], i[\rho,B] \rangle_{\rho,f}={\rm
Cov}_{\rho}(A,B)-{\rm Tr}(m_{\tilde f}(L_{\rho},R_{\rho})(A_0)B_0).
$$
\end{theorem}

As a consequence of the spectral theorem and of Theorem \ref{!} one
has the following relations.

\begin{proposition}{\rm \cite{GibiliscoImparatoIsola:2007}} \label{?}
    Let $\left\{\varphi_h\right\}$ be a complete orthonormal base
    composed of eigenvectors of $\rho$, and $\{ {\lambda}_h \}$ be the
    corresponding eigenvalues.  To self-adjoint matrices $A$, $B$ we
    associate matrices $a=a(\rho)$, $b=b(\rho)$ whose entries are
    given respectively by $a_{hj} \equiv \langle {A_0} {\varphi}_h
    |{\varphi}_j \rangle $, $b_{hj} \equiv \langle B_0 \varphi_h |
    {\varphi_j } \rangle $.  We have the following identities.
    \begin{align*}
	{\rm Cov}_{\rho} (A,B)&= {\rm Re} \{{\rm Tr} (\rho A_0 B_0)
	\}= \frac{1}{2} \sum_{h,j}({\lambda}_h+{\lambda}_j) {\rm Re}
	\{a_{hj}b_{jh} \} \\
	\frac{f(0)}{2}\langle i[\rho,A], i[\rho,B]\rangle_{\rho,f} &=
	\frac{1}{2} \sum_{h,j}({\lambda}_h+{\lambda}_j) {\rm Re} \{
	a_{hj} b_{jh} \} - \sum_{h,j} m_{\tilde
	f}(\lambda_h,\lambda_j) {\rm Re} \{ {a_{hj} b_{jh} } \} .
    \end{align*}
\end{proposition}

In what follows, capital letters will denote self-adjoint matrices and
the corresponding lower-case letters will be used for the above
transformation.

We also need the following result.

\begin{proposition} {\rm (\cite{GibiliscoImparatoIsola:2007}, Corollary
11.5) }\label{puri}

On pure states
$$
{\rm Tr}(m_{\tilde f}(L_{\rho},R_{\rho})(A_0)B_0)=0.
$$
\end{proposition}

\section{The $N$-volume conjectures for quantum Fisher informations}
\label{con}

Let $(V, g(\cdot,\cdot))$ be a real inner-product vector space.  By
$\langle u,v \rangle$ we denote the standard scalar product for
vectors $u,v \in {\mathbb R}^N$.

\begin{proposition}
    Let $v_1, ...  , v_N \in V$.  The real $N \times N$ matrix
    $G:=\{g(v_h,v_j) \}$ is positive semidefinite and therefore ${\rm
    det}\{g(v_h,v_j) \} \geq 0$.
\end{proposition}
\begin{proof}
    Let $x:=(x_1,...,x_N)\in {\mathbb R}^N$.  We have
    $$
    0 \leq g\big(\sum_j x_jv_j, \sum_j x_jv_j \big)=\sum_{h,j}
    x_hx_jg(v_h,v_j)= \langle x, G(x) \rangle.
    $$
\end{proof}

Motivated by the case $(V, g(\cdot,\cdot))=({\mathbb R}^N, \langle
\cdot, \cdot \rangle)$ one can give the following definition.

\begin{definition}
    $$
    {\rm Vol}^g(v_1, ... , v_N):=\sqrt{{\rm det}\{g(v_h,v_j) \}}.
    $$
\end{definition}

\begin{remark}
    \itm{i} Obviously, ${\rm Vol}^g(v_1, ...  , v_N) \geq 0$, where
    the equality holds if and only if $v_1, ...  , v_N \in V$ are
    linearly dependent.

    \itm{ii} If the inner product depends on a further parameter so
    that $g(\cdot,\cdot)=g_{\rho}(\cdot,\cdot)$, we write ${\rm
    Vol}_{\rho}^g(v_1, ...  , v_N)={\rm Vol}^g(v_1, ...  , v_N)$.

    \itm{iii} In the case $(V,g_{\rho}(\cdot,\cdot))=({\cal
    L}^2_{\mathbb R}(\Omega, {\cal G},\rho), {\rm
    Cov}_{\rho}(\cdot,\cdot))$ the number ${\rm Vol}_{\rho}^{\rm
    Cov}(A_1, ...  , A_N)^2$ is also known as the {\it generalized
    variance} of the random vector $(A_1, ...  , A_N)$.
\end{remark}

In what follows we move to the noncommutative case.  Here $A_1,...A_N$
are self-adjoint matrices, $\rho$ is a (faithful) density matrix and
$g(\cdot,\cdot)={\rm Cov}_\rho(\cdot,\cdot)$ has been defined in
(\ref{cov}).  By ${\rm Vol}_{\rho}^f$ we denote the volume associated
to the quantum Fisher information $\langle
\cdot,\cdot\rangle_{\rho,f}$ given by the (regular) normalized
symmetric operator monotone function $f$.

Let $N \in {\mathbb N}$, $ f \in {\cal F}_{op}^{\, r}$, $\rho \in
{\cal D}^1_n$ and $A_1,...,A_N \in M_{n,sa}$ be arbitrary.  We
conjecture the following results.

\begin{conjecture} \label{main}
    \begin{equation}
	{\rm Vol}_{\rho}^{{\rm Cov}}(A_1, ...  , A_N) \geq \left(
	\frac{f(0)}{2}\right)^{\frac{N}{2}}{\rm
	Vol}_{\rho}^f(i[\rho,A_1], \ldots ,i[\rho,A_N]). 
	\label{con1}
    \end{equation}
\end{conjecture}

\begin{conjecture} \label{equality}

    The above inequality is an equality if and only if ${A_1}_0, ... 
    , {A_N}_0$ are linearly dependent.
\end{conjecture}

\begin{conjecture} \label{monot}

    Fix $N \in {\mathbb N}$, $\rho \in {\cal D}^1_n$ and $A_1,...,A_N
    \in M_{n,sa}$.  Given $f\in\mathcal F_{op}^{\, r}$, define
    $$
    V(f):=\left( \frac{f(0)}{2}\right)^{\frac{N}{2}}{\rm
    Vol}_{\rho}^f(i[\rho,A_1], ...  ,i[\rho,A_N]).
    $$
    Then, for any $f,g\in\mathcal F_{op}^{\, r}$
    $$
    {\tilde f} \leq {\tilde g} \quad \Longrightarrow \quad V(f) \geq
    V(g).
    $$
\end{conjecture}

\begin{remark} \label{portmanteau}

    \itm{i} Conjecture \ref{main} is equivalent to the following
    inequality
    $$
    {\rm det} \{ {\rm Cov}_{\rho}(A_h,A_j) \} \geq {\rm det} \left\{
    \frac{f(0)}{2} \langle i[\rho,A_h],i[\rho,A_j]\rangle_{\rho,f}
    \right\}.
    $$

    \itm{ii} If $\rho$ and $A_1,...,A_N$ are fixed, set
    $$
    F(f):={\rm det} \{ {\rm Cov}_{\rho}(A_h,A_j) \} - {\rm det}
    \left\{ \frac{f(0)}{2} \langle
    i[\rho,A_h],i[\rho,A_j]\rangle_{\rho,f} \right\}.
    $$
    Because of Theorem \ref{!} one has
    $$
    F(f)={\rm det} \{ {\rm Cov}_{\rho}(A_h,A_j) \} - {\rm det} \left\{
    {\rm Cov}_{\rho}(A_h,A_j)- {\rm Tr}(m_{\tilde
    f}(L_{\rho},R_{\rho})({A_h}_0){A_j}_0)\right\}.
    $$
    Therefore, Conjecture \ref{main} is equivalent to
    $$
    F(f) \geq 0.
    $$

    \itm{iii} Conjecture \ref{monot} is equivalent to
    $$
    \tilde{f} \leq \tilde{g} \qquad \Longrightarrow \quad F(f) \leq
    F(g).
    $$

    \itm{iv} \label{comment} Suppose that Conjecture \ref{main} is
    true.  One can prove the ``if" part of Conjecture \ref{equality}
    in the following way.  Since $(A_0)_0=A_0$ one has
    $$
    {\rm Cov}_{\rho}(A_1,A_2)={\rm Re}\{{\rm Tr}(\rho {A_1}_0
    {A_2}_0)\} ={\rm Cov}_{\rho}({A_1}_0, {A_2}_0).
    $$
    From this it follows
    $$
    {\rm Vol}_{\rho}^{{\rm Cov}}(A_1, ...  , A_N) ={\rm
    Vol}_{\rho}^{{\rm Cov}}({A_1}_0, ...  , {A_N}_0).
    $$
    Therefore, if ${A_1}_0, ...  , {A_N}_0$ are linearly dependent
    then
    $$
    0={\rm Vol}_{\rho}^{{\rm Cov}}({A_1}_0, ...  , {A_N}_0) ={\rm
    Vol}_{\rho}^{{\rm Cov}}(A_1, ...  , A_N) \geq \left(
    \frac{f(0)}{2}\right)^{\frac{N}{2}}{\rm Vol}_{\rho}^f(i[\rho,A_1],
    ...  ,i[\rho, A_N]) \geq 0
    $$
    and we are done.

    \itm{v} The inequality
    $$
    {\rm det} \{ {\rm Cov}_{\rho}(A_h,A_j) \} \geq {\rm det} \left\{
    {\rm Cov}_{\rho}(A_h,A_j)- {\rm Tr}(m_{\tilde
    f}(L_{\rho},R_{\rho})({A_h}_0){A_j}_0)\right\}
    $$
    makes sense also for not faithful states.
\end{remark}

Because of Proposition \ref{puri} one has (by an obvious extension of
the definition) the following result.

\begin{proposition}
    If $\rho$ is a pure state, then for any $ N \in {\mathbb N}, \quad
    f \in {\cal F}_{op}^{\, r}, \quad A_1, ...  , A_N \in M_{n,sa}$
    one has
    $$
    {\rm Vol}_{\rho}^{{\rm Cov}}(A_1, ...  , A_N) = \left(
    \frac{f(0)}{2}\right)^{\frac{N}{2}}{\rm Vol}_{\rho}^f(i[\rho,A_1],
    ...  ,i[\rho,A_N]).
    $$
\end{proposition}

\section{The length inequality} \label{L}

In this section we discuss the case $N=1$ of Conjectures \ref{main},
\ref{equality} and \ref{monot}.  The cases $f=f_{SLD}$ and $f=f_{WY}$
of Conjecture \ref{main} were proved by Luo in \cite{Luo:2000} and
\cite{Luo:2003b}.  The general case of Conjecture \ref{main} was
proved by Hansen in \cite{Hansen:2006b} and shortly after by
Gibilisco, Imparato and Isola with a different technique in
\cite{GibiliscoImparatoIsola:2007}.  Conjectures \ref{equality} and
\ref{monot} have been proved by Gibilisco, Imparato and Isola in
\cite{GibiliscoImparatoIsola:2007} (see also
\cite{GibiliscoImparatoIsola:2007a}).

The proof of Conjecture \ref{main} by Hansen is based on the following
immediate proposition.

\begin{proposition}

    Let $T,S$ be real functions on the state space coinciding on pure
    states.  Suppose that $T$ is convex and $S$ is concave.  Then for
    all states $\rho$
    $$
    T(\rho) \leq S(\rho).
    $$
\end{proposition}

It is well known that the variance is concave.  Hansen was able to
prove that the {\it metric adjusted skew information} (namely
$\frac{f(0)}{2}{\rm Vol}_{\rho}^f(i[\rho,A_1])^2$) is convex and so he
got the conclusion from the above Proposition.  Note that the
convexity of the function $\frac{f(0)}{2}{\rm
Vol}_{\rho}^f(i[\rho,A_1])^2$ is related to the well known Lieb's
concavity theorem (see \cite{Hansen:2006}\cite{Hansen:2006b}). 
Despite the elegance of the above proof its ideas do not apply to
cases different from $N=1$, as we shall see in the next section.

The techniques applied by ourselves in the proof of case $N=1$ in the
paper \cite{GibiliscoImparatoIsola:2007} do not seem to share the same
fate.  Moreover they allow one to prove also Conjectures
\ref{equality} and \ref{monot}.  Let us discuss them.

\begin{theorem}
    Conjectures \ref{main}, \ref{equality} and \ref{monot} are true
    for $N=1$ and for any $f \in {\cal F}_{op}^r$.
\end{theorem}
\begin{proof}
    Set $A_1=A$.  Using Proposition \ref{?} and the notation in 
    Remark \ref{portmanteau} $(ii)$ one gets:
    \begin{align*}
	F(f)&={\rm det} \{ {\rm Cov}_{\rho}(A,A) \} - {\rm det}
	\left\{ {\rm Cov}_{\rho}(A,A)- {\rm Tr}(m_{\tilde
	f}(L_{\rho},R_{\rho})({A}_0){A}_0)\right\} \\
	&= {\rm Cov}_{\rho}(A,A) - [{\rm Cov}_{\rho}(A,A)- {\rm
	Tr}(m_{\tilde f}(L_{\rho},R_{\rho})({A}_0){A}_0)] \\
	&={\rm Tr}(m_{\tilde f}(L_{\rho},R_{\rho})({A}_0){A}_0) \\
	&=\sum_{i,j}m_{\tilde f}(\lambda_i,\lambda_j){\rm Re}\{
	a_{ij}a_{ji}\} \\
	&=\sum_{i,j}m_{\tilde f}(\lambda_i,\lambda_j) |a_{ij}|^2 \geq
	0,
    \end{align*}
    that is, Conjecture \ref{main} is true.  Obviously $F(f)=0$ iff
    $a_{ij}=0$ $\forall i,j$, that is, iff $A_0=0$ and so we get
    Conjecture \ref{equality}.  Using Proposition \ref{min} and Remark
    \ref{portmanteau} $(iii)$ one obtains also the validity of
    Conjecture \ref{monot}.
\end{proof}

\section{The area inequality} \label{A}

Let us discuss the case $N=2$ of Conjecture \ref{main}.  The result
was proved true for $f=f_{WY}$ by Luo, Q. Zhang and Z. Zhang in
\cite{LuoZZhang:2004} \cite{LuoQZhang:2004} \cite{LuoQZhang:2005}. 
The case $f=f_{WYD(\beta)}$, $\beta \in (0,\frac{1}{2})$ was proved by
Kosaki in \cite{Kosaki:2005} and shortly after by
Yanagi-Furuichi-Kuriyama in \cite{YanagiFuruichiKuriyama:2005}.  The
general case is due to Gibilisco, Imparato and Isola (see
\cite{GibiliscoIsola:2007} \cite{GibiliscoImparatoIsola:2007}).

Conjectures \ref{equality} and \ref{monot} were proved true by Kosaki
for the particular case $f=f_{WYD(\beta)}$.  The general case was
solved by Gibilisco, Imparato and Isola (see
\cite{GibiliscoIsola:2007} \cite{GibiliscoImparatoIsola:2007}).

First of all, let us show that the ideas used by Hansen in the case
$N=1$ do not apply to the case $N=2$. The problem is the lack of
concavity (and  convexity) for the generalized variance. We were not
able to find a counterexample in the literature, so we provide here
the simplest we found.

Let $\Omega:=\{1,2,...,n\}$. The space of (faithful) probability
measures on $\Omega$ is
$$
{\cal P}_n^1 :=\big\{ \rho \in {\mathbb R}^n \vert \sum \rho_i=1, \,
\rho_i > 0 \big\}.
$$
Let $X,Y \in {\mathbb R}^n$ be fixed random variables on $\Omega$.

\begin{proposition}
    The function $S:{\cal P}_n^1 \to {\mathbb R}$ given by
    \[
    S(\rho):=\Var_\rho(X)\Var_\rho(Y)-\Cov_\rho(X,Y)^2
    \] 
    is neither a concave nor a convex function.
\end{proposition}
\begin{proof}
    Let us compute the Hessian matrix $H^{XY}(\rho)$ of $S$ at the
    point $\rho$:
    \[
    \begin{array}{rcl}
	H_{ij}^{XY}(\rho)&=&\Var_\rho(Y)
	\frac{\de^{2}}{\de\r_{i}\de\r_{j}} \Var_\rho(X)+
	\frac{\de}{\de\r_{i}} \Var_\rho(X) \frac{\de}{\de\r_{j}}
	\Var_\rho(Y)\\[12pt] &+&\Var_\rho(X)
	\frac{\de^{2}}{\de\r_{i}\de\r_{j}} \Var_\rho(Y) +
	\frac{\de}{\de\r_{j}} \Var_\rho(X) \frac{\de}{\de\r_{i}}
	\Var_\rho(Y)\\[12pt] &-& 2 \frac{\de}{\de\r_{i}}
	\Cov_\rho(X,Y) \frac{\de}{\de\r_{j}} \Cov_\rho(X,Y) -
	2\Cov_\rho(X,Y) \frac{\de^{2}}{\de\r_{i}\de\r_{j}}
	\Cov_\rho(X,Y).
    \end{array}
    \] 
    If $X=(x_1,....,x_n), Y=(y_1,...,y_n)$, an explicit computation
    shows that
    \[
    \begin{array}{rcl}
	\frac{\de}{\de\r_{i}} \Cov_\rho(X,Y)&=&x_iy_i-x_i\mathbb
	E_\rho[Y]-y_i\mathbb E_\rho[X] \\[12pt]
	\frac{\de^{2}}{\de\r_{i}\de\r_{j}}
	\Cov_\rho(X,Y)&=&-x_iy_j-y_ix_j,
    \end{array}
    \] 
    so that
    \begin{equation}
	\begin{array}{rcl}
	    H_{ij}^{XY}(\rho)&=&-2x_ix_j\Var_\rho(Y)+(x_i^2-2x_i\mathbb
	    E_\rho[X])(y_j^2-2y_j\mathbb
	    E_\rho[Y])\\[12pt]&-&2y_iy_j\Var_\rho(X)+(y_i^2-2y_i\mathbb
	    E_\rho[Y])(x_j^2-2x_j\mathbb
	    E_\rho[X])\\[12pt]&-&2(x_iy_i-x_i\mathbb
	    E_\rho[Y]-y_i\mathbb E_\rho[X])(x_jy_j-x_j\mathbb
	    E_\rho[Y]-y_j\mathbb E_\rho[X])\\[12pt]&+&2\Cov_\rho(X,Y)
	    (x_iy_j+y_ix_j).
	\end{array}\label{h}
    \end{equation}

    In order to prove that in general $H_{ij}^{XY}(\rho)$ is neither
    negative semidefinite nor positive semidefinite (that is,
    $S(\rho)$ is neither concave nor convex) let $n=3$ and
    $\rho=(\frac{1}{3},\frac{1}{3},\frac{1}{3})$ be the uniform
    distribution, $X=(1,0,-1)^T$ and $Y=(1,-2,1)^T$.  Then $\mathbb
    E_\rho[X]=\mathbb E_\rho[Y]=\mathbb E_\rho[XY]=0$ and
    $\Cov_\rho(X,Y)=0$, so that (\ref{h}) reduces to
    \[
    H_{ij}^{XY}(\rho) = -2x_ix_j\Var_\rho(Y)+x_i^2y_j^2-2y_iy_j
    \Var_\rho(X) + x_j^2y_i^2-2x_ix_jy_iy_j.
    \] 
    Hence, given $\alpha\in\mathbb R^3$
    \[
    \alpha^TH\alpha = -2\left[\Var_\rho(Y)
    \left(\sum_ix_i\alpha_i\right)^2 -
    \sum_ix_i^2\alpha_i\sum_iy_i^2\alpha_i + \Var_\rho(X)
    \left(\sum_iy_i\alpha_i\right)^2 +
    \left(\sum_ix_iy_i\alpha_i\right)^2\right].
    \] 
    In particular, $\alpha=\rho$ implies
    \[\rho^TH\rho=2\Var_\rho(X)\Var_\rho(Y)> 0,\] while
    $\alpha=(0,\alpha_2,0)$, $\alpha_2\neq 0$, implies
    $\alpha^TH\alpha<0$.
\end{proof}

Now we describe how the ideas for the proof of the length inequality $(N=1)$
can be modified to apply to the case of the area inequality ($N=2$).

\begin{definition}
    For any $f \in {\cal F}_{op}^{\, r}$ set
    $$
    H^{ f}(x,y,w,z):=\frac12(x+y)m_{\tilde
    f}(w,z)+\frac12(w+z)m_{\tilde f}(x,y)-m_{\tilde f}(x,y)m_{\tilde
    f}(w,z), \qquad \qquad x,y,w,z > 0.
    $$
    Given $\rho\in\mathcal D_n^1$ and $\{\lambda_i\}$, $i=1\ldots, n$,
    the corresponding eigenvalues, we set
    \[H_{ijkl}^f:=H^f(\lambda_i,\lambda_j,\lambda_k,\lambda_l).\]
\end{definition}

\begin{proposition} \label{monarea}{\rm \cite{GibiliscoImparatoIsola:2007}}

    For any $f,g \in {\cal F}_{op}^{\, r}$ and for any $x,y,w,z >0$ one has:
    $$
    \tilde{f} \leq \tilde{g} \qquad \Longrightarrow \qquad 0\leq
    H^{f}(x,y,w,z) \leq H^g(x,y,w,z).
    $$
\end{proposition}

Using the same notations as in Proposition \ref{?}, one can give the
following definition.

\begin{definition}
    Set
    $$
    K_{i,j,k,l}:= K_{i,j,k,l}(\rho,A,B):= |a_{ij} |^2 |b_{kl}|^2 +
    |a_{kl} |^2 |b_{ij}|^2 - 2{\rm Re} \{ a_{ij} b_{ji} \} {\rm Re} \{
    a_{kl} b_{lk} \} .
    $$
\end{definition}

Note that $K_{i,j,k,l}$ does not depend on $f$.  Since
$$
|a_{ij} |^2 |b_{kl}|^2 + |a_{kl} |^2 |b_{ij}|^2 \ge 2\left| {a_{ij}
b_{ji} } \right|\left| {a_{kl} b_{lk} } \right| \ge 2\left|
{{\mathop{\rm Re}\nolimits} \left\{ {a_{ij} b_{ji} }
\right\}{\mathop{\rm Re}\nolimits} \left\{ {a_{kl} b_{lk} } \right\}}
\right|,
$$
we get that $K_{ijkl}$ is non-negative. Moreover one has the following result.

\begin{proposition} {\rm \cite{GibiliscoImparatoIsola:2007}} \label{K=0}
    $K_{i,j,k,l} = 0, \quad \forall i,j,k,l$ $\Longleftrightarrow$
    $A_0,B_0$ are linearly dependent.
\end{proposition}

Recall that
$$
F(f)={\rm det} \{ {\rm Cov}_{\rho}(A_i,A_j) \} - {\rm det} \left\{
{\rm Cov}_{\rho}(A_i,A_j)- {\rm Tr}(m_{\tilde
f}(L_{\rho},R_{\rho})({A_i}_0){A_j}_0)\right\}.
$$

\begin{theorem} {\rm \cite{GibiliscoImparatoIsola:2007} }
    For $N=2$ one has: 
    \[
    F(f)=\dfrac1{2}\sum_{i,j,k,l}H_{i,j,k,l}^f \cdot K_{i,j,k,l}.
    \]
\end{theorem}

>From the above Theorem one gets the following result.

\begin{theorem}
    Conjectures \ref{main}, \ref{equality} and \ref{monot} are true
    for $N=2$.
\end{theorem}
\begin{proof}
    Since $H_{i,j,k,l} > 0$ and $K_{i,j,k,l} \geq 0$ we get $F(f) \geq
    0$ and therefore Conjecture \ref{main} is true.

    From Proposition \ref{K=0} we get that $F(f)=0$ iff $A_0,B_0$ are
    linearly dependent, that is, Conjecture \ref{equality} holds.

    From Proposition \ref{monarea} we get that $\tilde{f} \leq
    \tilde{g}$ implies $F(f) \leq F(g)$ and, therefore, one proves
    Conjecture \ref{monot}.
\end{proof}

\section{The main result: the volume inequality} \label{core}

In this section we study the case $N=3$ of Conjectures \ref{main},
\ref{equality} and \ref{monot}.  In the sequel we need the following
function.

\begin{definition}
    For any $f \in {\cal F}_{op}^{\, r}$, set
    $$
    \begin{array}{rcl}
	H^{ f}(x,y,h,k,w,z)&:=&\frac14(x+y)(h+k)m_{\tilde
	f}(w,z)+\frac14(w+z)(x+y)m_{\tilde f}(h,k)+
	\frac14(h+k)(w+z)m_{\tilde f}(x,y)\\[12pt]
	&&-\frac12(x+y)m_{\tilde f}(h,k)m_{\tilde f}(w,z)
	-\frac12(w+z)m_{\tilde f}(x,y)m_{\tilde
	f}(h,k)-\frac12(h+k)m_{\tilde f}(w,z)m_{\tilde f}(x,y)
	\\[12pt] &&+m_{\tilde f}(x,y)m_{\tilde f}(h,k)m_{\tilde
	f}(w,z),\ \ \ \ x,y,h,k,w,z>0.
    \end{array}
    $$
\end{definition}

\begin{proposition} \label{vol} 
    For any $f\in {\cal F}_{op}^{\, r}$ one has $H^f(x,y,h,k,w,z)> 0$.
\end{proposition}
\begin{proof}
    $$
    \begin{array}{rcl}
	H^{ f}(x,y,h,k,w,z)&=&\frac12(x+y)m_{\tilde
	f}(w,z)[\frac12(h+k)- m_{\tilde f}(h,k)] \\[12pt]
	&&+\frac12(w+z)m_{\tilde f}(h,k)[\frac12(x+y)-m_{\tilde
	f}(x,y)]\\[12pt]&& +\frac12(h+k)m_{\tilde f}(x,y)
	[\frac12(w+z) -m_{\tilde f}(w,z)]\\[12pt] &&+m_{\tilde
	f}(x,y)m_{\tilde f}(h,k)m_{\tilde f}(w,z)\\[12pt]&&> 0,
    \end{array}
    $$
    since, for any $f\in\mathcal F_{op}^{\, r}$, $m_{\tilde f}$ is
    smaller than the arithmetic mean.
\end{proof}

\begin{proposition}\label{mon} 
    For any $x,y,h,k,w,z>0$
    $$
    \tilde{f} \leq \tilde{g} \qquad \Longrightarrow \qquad
    H^{f}(x,y,h,k,w,z) \leq H^g(x,y,h,k,w,z).
    $$
\end{proposition}
\begin{proof}
    Since
    $$
    \frac{x+y}2-m_{\tilde f}(x,y)= \frac{(x-y)^2}{2y} \cdot
    \frac{f(0)}{f(\frac{x}{y}) },
    $$
    we have
    \[
    \begin{array}{rcl}
	H^{ f}(x,y,h,k,w,z)&=&\frac12(x+y)m_{\tilde
	f}(w,z)[\frac12(h+k)- m_{\tilde f}(h,k)]
	\\[12pt]&&+\frac12(w+z)m_{\tilde
	f}(h,k)[\frac12(x+y)-m_{\tilde f}(x,y)]\\[12pt]&&
	+\frac12(h+k)m_{\tilde f}(x,y) [\frac12(w+z) -m_{\tilde
	f}(w,z)]\\[12pt] &&+m_{\tilde f}(x,y)m_{\tilde
	f}(h,k)m_{\tilde f}(w,z)
	\\[12pt]&=&\displaystyle\frac12(x+y)\left(\dfrac{w+z}2-\dfrac{(w-z)^2}{2z}
	\dfrac{f(0)}{f(\frac{w}{z})}\right)\dfrac{(h-k)^2}{2k}\dfrac{f(0)}{f(\frac{h}{k})}
	\\[12pt]&&+\displaystyle\frac12(w+z)\left(\dfrac{h+k}2-\dfrac{(h-k)^2}{2k}
	\dfrac{f(0)}{f(\frac{h}{k})}\right)\dfrac{(x-y)^2}{2y}\dfrac{f(0)}{f(\frac{x}{y})}
	\\[12pt]&&+\displaystyle\frac12(h+k)\left(\dfrac{x+y}2-\dfrac{(x-y)^2}{2y}
	\dfrac{f(0)}{f(\frac{x}{y})}\right)\dfrac{(w-z)^2}{2z}\dfrac{f(0)}{f(\frac{w}{z})}\\[12pt]
	&&+\left(\dfrac{x+y}2-\dfrac{(x-y)^2}{2y}
	\dfrac{f(0)}{f(\frac{x}{y})}\right)\left(\dfrac{h+k}2-\dfrac{(h-k)^2}{2k}
	\dfrac{f(0)}{f(\frac{h}{k})}\right)\left(\dfrac{w+z}2-\dfrac{(w-z)^2}{2z}
	\dfrac{f(0)}{f(\frac{w}{z})}\right)\\[12pt]&=&
	\dfrac18\left((x+y)(h+k)(w+z)-\dfrac{(x-y)^2(h-k)^2(w-z)^2}{ykz}\dfrac{f(0)}{f(\frac{x}{y})}\dfrac{f(0)}{f(\frac{h}{k})}
	\dfrac{f(0)}{f(\frac{w}{z})}\right).
    \end{array}
    \]

    Proposition \ref{min} states that
    $$
    \tilde{f} \leq \tilde{g} \Rightarrow \frac{f(0)}{f(t)} \geq
    \frac{g(0)}{g(t)}, \qquad \qquad \forall t >0;
    $$
    hence, we obtain
    $$
    H^{ f}(x,y,h,k,w,z) \leq H^g(x,y,h,k,w,z), \qquad \qquad \forall
    x,y,h,k,w,z >0
    $$
    by elementary computations.
\end{proof}

Now we use again the definitions given in Proposition \ref{?}.  Let
$\left\{\varphi_i\right\}$ be a complete orthonormal base composed of
eigenvectors of $\rho$, and $\{ {\lambda}_i \}$ the corresponding
eigenvalues.  For any $A,B,C\in\mathbb M_{n,sa}$, set $a_{ij} \equiv
\langle {A_0} {\varphi}_i |{\varphi}_j \rangle $, $ b_{ij} \equiv
\langle B_0 \varphi_i | {\varphi_j } \rangle $ and $ c_{ij} \equiv
\langle C_0 \varphi_i | {\varphi_j } \rangle $.

In order to state the following result, fix $i,j,h,k,l,m \in
\{1,...,n\}$.  Define
$$
\alpha(1)=i, \quad \alpha(2)=h, \quad \alpha(3)=l;
$$
$$
\beta(1)=j, \quad \beta(2)=k, \quad \beta(3)=m.
$$ 

Let $S_3$ be the symmetric group on 3 elements and let $A_3\subset
S_3$ be the alternating group.  If $\sigma \in S_3$ and $i=1,2,3$
define $\sigma_i:=\sigma(i)$.

\begin{definition}
    \begin{equation}
	\begin{array}{rcl}
	    H_{ijhklm}^f&:=&\displaystyle
	    H^f(\lambda_i,\lambda_j,\lambda_h,\lambda_k,\lambda_l,\lambda_m),\\
	    K_{ijhklm}&:=&\displaystyle\sum_{\sigma\in S_3} \bigg{[}
	    |a_{\alpha(\sigma_1)\beta(\sigma_1)}|^2|
	    b_{\alpha(\sigma_2)\beta(\sigma_2)}|^2|c_{\alpha(\sigma_3)\beta(\sigma_3)}|^2\\[12pt]&&
	    +2{\rm Re} \{ a_{\alpha(\sigma_1)\beta(\sigma_1) }
	    c_{\beta(\sigma_1)\alpha(\sigma_1)} \} {\rm Re} \{
	    a_{\alpha(\sigma_2)\beta(\sigma_2)}
	    b_{\beta(\sigma_2)\alpha(\sigma_2)} \}{\rm Re} \{
	    b_{\alpha(\sigma_3)\beta(\sigma_3)}
	    c_{\beta(\sigma_3)\alpha(\sigma_3)} \} \bigg{]}
	    \\[12pt]&&-2 \displaystyle\sum_{\sigma\in A_3} \bigg{[}
	    |a_{\alpha(\sigma_1)\beta(\sigma_1)}|^2{\rm Re} \{
	    b_{\alpha(\sigma_2)\beta(\sigma_2)}
	    c_{\beta(\sigma_2)\alpha(\sigma_2)} \} {\rm Re} \{
	    b_{\alpha(\sigma_3)\beta(\sigma_3)}
	    c_{\beta(\sigma_3)\alpha(\sigma_3)}\}\\[12pt]&& +
	    |b_{\alpha(\sigma_1)\beta(\sigma_1)}|^2{\rm Re} \{
	    c_{\alpha(\sigma_2)\beta(\sigma_2)}
	    a_{\beta(\sigma_2)\alpha(\sigma_2)} \} {\rm Re} \{
	    c_{\alpha(\sigma_3)\beta(\sigma_3)}
	    a_{\beta(\sigma_3)\alpha(\sigma_3)}\}\\[12pt]&&+
	    |c_{\alpha(\sigma_1)\beta(\sigma_1)}|^2{\rm Re} \{
	    a_{\alpha(\sigma_2)\beta(\sigma_2)}
	    b_{\beta(\sigma_2)\alpha(\sigma_2)} \} {\rm Re} \{
	    a_{\alpha(\sigma_3)\beta(\sigma_3)}
	    b_{\beta(\sigma_3)\alpha(\sigma_3)}\} \bigg{]}.
	\end{array}\label{k}
    \end{equation}
\end{definition}

\begin{lemma} \label{lem}
    When $N=3$,
    \[
    F(f)=\dfrac16\sum_{i,j,h,k,l,m}H^f_{ijhklm}K_{ijhklm}.
    \]
\end{lemma}
\begin{proof}
    From Proposition \ref{?}, for any $A,B\in M_{n,sa}$
    \begin{align*}
	\Var_{\rho}(A) &= {\rm Tr} ( \rho A_0^2) =
	\frac{1}{2}\sum_{i,j} ( \lambda _i + \lambda_j ) a_{ij} a_{ji}
	\\
	{\rm Cov}_{\rho} (A,B)&= {\rm Re} \{{\rm Tr} (\rho A_0 B_0)
	\}= \frac{1}{2} \sum_{i,j}({\lambda}_i+{\lambda}_j) {\rm Re}
	\{a_{ij}b_{ji} \} \\
        \frac{f(0)}{2}\langle i[\rho,A], i[\rho,B]\rangle_{\rho,f} &=
	\frac{1}{2} \sum_{i,j}({\lambda}_i+{\lambda}_j) {\rm Re} \{
	a_{ij} b_{ji} \} - \sum_{i,j} m_{\tilde
	f}(\lambda_i,\lambda_j) {\rm Re} \{ {a_{ij} b_{ji} } \} .
    \end{align*}
    Set
    \[
    \begin{array}{rcl}
	\xi &:=& \hbox{Var}_\rho \left( A \right) \hbox{Var}_\rho
	\left( B \right)\hbox{Var}_\rho\left(C\right) -
	\dfrac{f(0)^3}{8}||i[\rho,A]||^2_{\rho,f} \cdot
	||i[\rho,B]||^2_{\rho,f}\cdot ||i[\rho,C]||^2_{\rho,f}
	\\[12pt] &=&\dfrac{1}{2}\displaystyle\sum_{i,j,h,k,l,m}
	\left\{\frac12(\lambda_i+\lambda_j)(\lambda_h+\lambda_k)m_{\tilde
	f}(\lambda_l,\lambda_m)+\dfrac12(\lambda_l+\lambda_m)(\lambda_i+\lambda_j)
	m_{\tilde f}(\lambda_h,\lambda_k)\right.\\[12pt]
	&&+\displaystyle
	\dfrac12(\lambda_h+\lambda_k)(\lambda_l+\lambda_m)m_{\tilde
	f}(\lambda_i,\lambda_j) -(\lambda_i+\lambda_j)m_{\tilde
	f}(\lambda_h,\lambda_k)m_{\tilde f}(\lambda_l,\lambda_m)
	-(\lambda_l+\lambda_m)m_{\tilde
	f}(\lambda_i,\lambda_j)m_{\tilde
	f}(\lambda_h,\lambda_k)\\[12pt]
	&&-\displaystyle(\lambda_h+\lambda_k)m_{\tilde
	f}(\lambda_l,\lambda_m)m_{\tilde
	f}(\lambda_i,\lambda_j)+2m_{\tilde
	f}(\lambda_i,\lambda_j)m_{\tilde
	f}(\lambda_h,\lambda_k)m_{\tilde
	f}(\lambda_l,\lambda_m)\bigg\} a_{ij} a_{ji} b_{hk}
	b_{kh}c_{lm}c_{ml}\\[12pt]& = &\displaystyle\sum_{i,j,h,k,l,m}
	H_{ijhklm}^f|a_{ij}|^2|b_{hk}|^2|c_{lm}|^2
	\\[12pt]&=&\dfrac1{6}\displaystyle\sum_{i,j,h,k,l,m}
	H_{ijhklm}^f \big\{ |a_{ij}|^2|b_{hk}|^2|c_{lm}|^2 +
	|a_{hk}|^2|b_{ij}|^2|c_{lm}|^2 +
	|a_{ij}|^2|b_{lm}|^2|c_{hk}|^2\\[12pt] &&
	+|a_{lm}|^2|b_{hk}|^2|c_{ij}|^2 +
	|a_{hk}|^2|b_{lm}|^2|c_{ij}|^2 +
	|a_{lm}|^2|b_{ij}|^2|c_{hk}|^2\big\},
    \end{array}
    \]
    \[
    \begin{array}{rcl}
	\gamma &:=& {\rm Cov}_{\rho}(A,B){\rm Cov}_{\rho}(A,C){\rm
	Cov}_{\rho}(B,C)-\dfrac{f(0)^3}{8}\langle i[\rho,A],
	i[\rho,B]\rangle_{\rho,f}i[\rho,A],i[\rho,C]\rangle_{\rho,f}i[\rho,B],
	i[\rho,C]\rangle_{\rho,f} \\[12pt]
	&=&\dfrac{1}{2}\displaystyle\sum_{i,j,h,k,l,m}
	\bigg\{\frac12(\lambda_i+\lambda_j)(\lambda_h+\lambda_k)m_{\tilde
	f}(\lambda_l,\lambda_m)+\dfrac12(\lambda_l+\lambda_m)(\lambda_i+\lambda_j)m_{\tilde
	f}(\lambda_h,\lambda_k)\\[12pt]&&+\displaystyle
	\dfrac12(\lambda_h+\lambda_k)(\lambda_l+\lambda_m)m_{\tilde
	f}(\lambda_i,\lambda_j) -(\lambda_i+\lambda_j)m_{\tilde
	f}(\lambda_h,\lambda_k)m_{\tilde f}(\lambda_l,\lambda_m)
	-(\lambda_l+\lambda_m)m_{\tilde
	f}(\lambda_i,\lambda_j)m_{\tilde
	f}(\lambda_h,\lambda_k)\\[12pt]&&-\displaystyle(\lambda_h+\lambda_k)m_{\tilde
	f}(\lambda_l,\lambda_m)m_{\tilde
	f}(\lambda_i,\lambda_j)+2m_{\tilde
	f}(\lambda_i,\lambda_j)m_{\tilde
	f}(\lambda_h,\lambda_k)m_{\tilde
	f}(\lambda_l,\lambda_m)\bigg\} {\rm Re} \{ a_{ij} b_{ji} \}
	{\rm Re} \{ a_{hk} c_{kh} \}{\rm Re} \{ b_{lm} c_{ml} \}
	\\[12pt]&=&\displaystyle\sum_{i,j,h,k,l,m} H_{ijhklm}^f {\rm
	Re} \{ a_{ij} b_{ji} \} {\rm Re} \{ a_{hk} c_{kh} \}{\rm Re}
	\{ b_{lm} c_{ml} \}
	\\[12pt]&=&\dfrac1{6}\displaystyle\sum_{i,j,h,k,l,m}
	H_{ijhklm}^f\big\{{\rm Re} \{ a_{ij} b_{ji} \} {\rm Re} \{
	a_{hk} c_{kh} \}{\rm Re} \{ b_{lm} c_{ml} \} +{\rm Re} \{
	a_{ij} c_{ji} \} {\rm Re} \{ a_{hk} b_{kh} \}{\rm Re} \{
	b_{lm} c_{ml} \}\\[12pt]&&+ {\rm Re} \{ a_{ij} b_{ji} \} {\rm
	Re} \{ b_{hk} c_{kh} \}{\rm Re} \{ a_{lm} c_{ml} \} +{\rm Re}
	\{ a_{hk} b_{kh} \} {\rm Re} \{ a_{lm} c_{ml} \}{\rm Re} \{
	b_{ij} c_{ji}\}\\[12pt]&&+{\rm Re} \{ a_{ij} c_{ji} \} {\rm
	Re} \{ a_{lm} b_{ml} \}{\rm Re} \{ b_{hk} c_{kh} \} \}+{\rm
	Re} \{ a_{hk} c_{kh} \} {\rm Re} \{ a_{lm} b_{ml} \}{\rm Re}
	\{ b_{ij} c_{ji} \}\big\},
    \end{array} 
    \]
    \[
    \begin{array}{rcl}
	\eta_1 &:=& \hbox{Var}_\rho(A)|{\rm Cov}_{\rho}(B,C)|^2
	-\dfrac{f(0)^3}{8} ||i[\rho,A]||^2_{\rho,f} | \langle
	i[\rho,B], i[\rho,C]\rangle^2_{\rho,f}|^2 \\[12pt]
	&=&\dfrac{1}{2}\displaystyle\sum_{i,j,h,k,l,m}
	\bigg\{\frac12(\lambda_i+\lambda_j)(\lambda_h+\lambda_k)m_{\tilde
	f}(\lambda_l,\lambda_m)+\dfrac12(\lambda_l+\lambda_m)(\lambda_i+\lambda_j)m_{\tilde
	f}(\lambda_h,\lambda_k)\\[12pt]&&+\displaystyle
	\dfrac12(\lambda_h+\lambda_k)(\lambda_l+\lambda_m)m_{\tilde
	f}(\lambda_i,\lambda_j) -(\lambda_i+\lambda_j)m_{\tilde
	f}(\lambda_h,\lambda_k)m_{\tilde f}(\lambda_l,\lambda_m)
	-(\lambda_l+\lambda_m)m_{\tilde
	f}(\lambda_i,\lambda_j)m_{\tilde
	f}(\lambda_h,\lambda_k)\\[12pt]&&-\displaystyle(\lambda_h+\lambda_k)m_{\tilde
	f}(\lambda_l,\lambda_m)m_{\tilde
	f}(\lambda_i,\lambda_j)+2m_{\tilde
	f}(\lambda_i,\lambda_j)m_{\tilde
	f}(\lambda_h,\lambda_k)m_{\tilde
	f}(\lambda_l,\lambda_m)\bigg\} a_{ij}a_{ji}{\rm Re} \{ b_{hk}
	c_{kh} \} {\rm Re} \{ b_{lm} c_{ml} \}
	\\[12pt]&=&\displaystyle\sum_{i,j,h,k,l,m}
	H_{ijhklm}^f|a_{ij}|^2{\rm Re} \{ b_{hk} c_{hk} \} {\rm Re} \{
	b_{lm} c_{lm}\}
	\\[12pt]&=&\dfrac1{3}\displaystyle\sum_{i,j,h,k,l,m}H_{ijhklm}^f\big\{|a_{ij}|^2{\rm
	Re} \{ b_{hk} c_{kh} \} {\rm Re} \{ b_{lm} c_{ml}\}
	+|a_{hk}|^2{\rm Re} \{ b_{ij} c_{ji} \} {\rm Re} \{ b_{lm}
	c_{ml}\} \\[12pt]&&+|a_{lm}|^2{\rm Re} \{ b_{ij} c_{ji} \}
	{\rm Re} \{ b_{hk} c_{kh}\} \big\},
    \end{array}
    \]
    \[
    \begin{array}{rcl}
	\eta_2 &:=& \hbox{Var}_\rho(B)|{\rm Cov}_{\rho}(A,C)|^2
	-\dfrac{f(0)^3}{8} ||i[\rho,B]||^2_{\rho,f} | \langle
	i[\rho,A], i[\rho,C]\rangle^2_{\rho,f}|^2
	\\[12pt]&=&\dfrac1{3}\displaystyle\sum_{i,j,h,k,l,m}
	H_{ijhklm}^f\big\{|b_{ij}|^2{\rm Re} \{ a_{hk} c_{kh} \} {\rm
	Re} \{ a_{lm} c_{ml}\} +|b_{hk}|^2{\rm Re} \{ a_{ij} c_{ji} \}
	{\rm Re} \{ a_{lm} c_{ml}\} \\[12pt]&&+|b_{mn}|^2{\rm Re} \{
	a_{ij} c_{ji} \} {\rm Re} \{ a_{hk} c_{kh}
	\}\big\},
    \end{array}
    \]
    \[
    \begin{array}{rcl}
	\eta_3 &:=& \hbox{Var}_\rho(C)|{\rm Cov}_{\rho}(A,B)|^2
	-\dfrac{f(0)^3}{8} ||i[\rho,C]||^2_{\rho,f} | \langle
	i[\rho,A], i[\rho,B]\rangle^2_{\rho,f}|^2
	\\[12pt]&=&\dfrac1{3}\displaystyle\sum_{i,j,h,k,l,m}
	H_{ijhklm}^f\big\{|c_{ij}|^2{\rm Re} \{ a_{hk} b_{kh} \} {\rm
	Re} \{ a_{lm} b_{ml}\} +|c_{hk}|^2{\rm Re} \{ a_{ij} b_{ji} \}
	{\rm Re} \{ a_{lm} b_{ml}\} \\[12pt]&&+|c_{mn}|^2{\rm Re} \{
	a_{ij} b_{ji} \} {\rm Re} \{ a_{hk} b_{kh}
	\}\big\}.
    \end{array}
    \] 
    Then,
    \[
    \begin{array}{rcl}
	F(f)&=&\xi+2\gamma -\eta_1-\eta_2-\eta_3\\[12pt]&=&\dfrac16
	\displaystyle
	\sum_{i,j,k,l,m,n}H_{ijhklm}^fK_{ijhklm},
    \end{array}
    \] 
    where $K_{ijhklm}$ was defined in (\ref{k}).
\end{proof}

\begin{corollary}
    If $K_{ijhklm}\ge 0$ for any $i,j,h,k,l,m$, then
    \begin{description}
	\item[(i)]
	\begin{equation}
	    {\rm Vol}_{\rho}^{{\rm Cov}}(A,B,C) \geq \left(
	    \frac{f(0)}{2}\right)^{\frac{3}{2}}{\rm
	    Vol}_{\rho}^f(i[\rho,A], i[\rho,B],i[\rho,C]),\qquad \quad
	    \forall f \in {\cal F}_{op};\label{eq1}
	\end{equation}
	\item[(ii)] 
	\begin{equation}\tilde{f} \leq \tilde{g} \quad
	    \Longrightarrow \quad V(f) \geq V(g) \geq
	    0.\label{eq2}
	\end{equation}
    \end{description}\label{cor}
\end{corollary}
\begin{proof}
    $(i)$ is a direct consequence of Lemma \ref{lem}.  In order to
    prove $(ii)$, observe that because of Proposition \ref{vol} and
    Proposition \ref{mon} one has that
    $$
    \tilde{f} \leq \tilde{g} \quad \Longrightarrow \quad 0 <
    H_{ijhklm}^f \leq H_{ijhklm}^f.
    $$
    Since
    \[
    F(f)=\dfrac1{6}\sum_{i,j,h,k,l,m}H_{ijhklm}^fK_{ijhklm}
    \]
    and $K_{ijhklm}$ does not depend on f, we get
    \[
    \tilde{f} \leq \tilde{g} \quad \Longrightarrow \quad 0 \leq F(f)
    \leq F(g), 
    \] 
    so that the conclusion follows from the definition of $F$.
\end{proof}

\begin{proposition}
    Let $A,B,C\in\mathcal M_{n,sa}(\mathbb R)$; then $(\ref{eq1})$ and
    $(\ref{eq2})$ hold true.  Furthermore, equality in $(\ref{eq1})$
    holds if and only if $A_0,B_0,C_0$ are linearly independent.  This
    means that the case $N=3$ of the Conjectures \ref{main},
    \ref{equality} and \ref{monot} is true for real matrices.
\end{proposition}
\begin{proof}
    When $A,B,C$ are real matrices, then $K_{ijhklm}$ reduces to
    \begin{equation}
	\begin{array}{rcl}
	    K_{ijhklm}&=&(a_{ij}b_{hk}c_{lm}-a_{hk}b_{ij}c_{lm}-a_{ij}b_{lm}c_{hk}\\
	    [12pt] &&-a_{lm}b_{hk}c_{ij} + a_{hk}b_{lm}c_{ij}+
	    a_{lm}b_{ij}c_{hk})^2\ge 0,
	\end{array}\label{kr}
    \end{equation} 
    as one can verify by a direct computation.  Hence, Corollary
    \ref{cor} holds.  By definition of $a,$ $b$, $c$, observe that the
    hypothesis of linear independence of $A,B,C$ is equivalent to the
    linear independence of $a,b,c$.  If $A_0,B_0,C_0$ are linearly
    dependent, then by Remark \ref{comment} we are done.

    Conversely, suppose that $A,B,C$ are linearly independent; then we
    want to show that $F(f)>0$.  Since $H_{ijhklm}$ is positive and
    $K_{ijhklm}$ is non negative, this is equivalent to prove that
    $K_{ijhklm}$ is not null for some $i,j,h,k,l,m$.  By hypothesis,
    there exist $i,j,h,k,l,m$ such that the vectors
    $(a_{ij},a_{hk},a_{lm})$, $(b_{ij},b_{hk},b_{lm})$,
    $(c_{ij},c_{hk},c_{lm})$ are linearly independent.  This implies
    that $K_{ijhklm}>0$, since one can express $(\ref{kr})$ as
    \[
    K_{ijhklm}=\left[ \det \left(
    \begin{array}{ccc}
	a_{ij} & a_{hk} & a_{lm} \\
	b_{ij} & b_{hk} & b_{lm} \\
	c_{ij} & c_{hk} & c_{lm} \\
    \end{array}
    \right) \right]^2.
    \]
\end{proof}

\begin{proposition} \label{penultima}
    Let $\rho \in {\cal D}^1_n$ and $A,B,C\in\mathcal M_{n,sa}(\mathbb
    C)$ such that: $A$ is arbitrary; $b_{ii}=0$, $i=1\ldots, n$;
    $c_{ij}=0$, $i,j=1\ldots, n$ and $i\neq j$.  Then $K_{ijhklm}\ge
    0$ for any $i,j,h,k,l,m$, and therefore Conjectures \ref{main} and
    \ref{monot} are true.
\end{proposition}
\begin{proof} 
    Suppose $i\neq j$, $h\neq k$, $l\neq m$; then $K_{ijhklm}=0$.  If,
    say, $i=j$ and $h\neq k$, $l\neq m$, then
    \[
    \begin{array}{rcl}
	K_{ijhklm}&=&|a_{lm}|^2|b_{hk}|^2|c_{ii}|^2+|a_{hk}|^2|b_{lm}|^2|c_{ii}|^2-
	2|c_{ii}|^2{\rm Re}\{a_{hk}b_{kh}\}{\rm
	Re}\{a_{lm}b_{ml}\}\\[12pt]&\ge& (|a_{lm}||b_{hk}||c_{ii}|-
	|a_{hk}||b_{lm}||c_{ii}|)^2\ge 0.
    \end{array}
    \] 
    When $i=j$, $h=k$ and $l\neq m$, then
    \[
    \begin{array}{rcl}
	K_{ijhklm}&=&|a_{ii}|^2|b_{lm}|^2|c_{kk}|^2+|a_{kk}|^2|b_{lm}|^2|c_{ii}|^2-
	2|b_{lm}|^2{\rm Re}\{a_{ii}c_{ii}\}{\rm
	Re}\{a_{kk}c_{kk}\}\\[12pt]&\ge& (|a_{ii}||b_{lm}||c_{kk}|-
	|a_{kk}||b_{lm}||c_{ii}|)^2\ge 0.
    \end{array}
    \]
    Finally, $K_{iikkmm}=0$.  The remaining cases are similar to the
    previous ones.
\end{proof}

\begin{corollary} \label{comp}
    Let $\rho$ be a diagonal state, that is, $\rho_{ij}=0$ for $i\neq
    j$.  Let $A$, $B$, $C\in\mathcal M_n(\mathbb C)$ self-adjoint
    matrices such that: $A$ is arbitrary; $B_{ii}=0$, $i=1\ldots, n$;
    $C$ is diagonal.  Then $K_{ijhklm}\ge 0$ for any $i,j,h,k,l,m$ and
    therefore Conjectures \ref{main} and \ref{monot} are true.
\end{corollary}
\begin{proof} 
    Because of Proposition \ref{penultima} it is enough to prove that
    $B_{ii}=C_{ij}=0$ for $i\neq j$ implies $b_{ii}=c_{ij}=0$, $i\neq
    j$.  Indeed, since $C$ is diagonal, so is $C_0 := C - {\rm
    Tr}(\rho C)I$ and hence $c_{ij}=0$ for $i\not=j$.  On the other
    hand, if $B_{ii}=0$, $i=1,\ldots, n$, then
    \[
    {\rm Tr}(\rho B)=\sum_{k,l}\rho_{kl}B_{lk}=0,
    \] 
    since $\rho$ is diagonal, so that $B_0=B$.  This implies that
    $b_{ii}=B_{ii}=0$.
\end{proof}

\begin{remark}
    Examples of self-adjoint matrices with the above properties are
    given by Pauli matrices and by ge\-neralized Gell-Mann matrices
    (see \cite{BertiniCacciatoriCerchiai:2006} and references
    therein).
\end{remark}

\section{Covariance, correlation and entanglement} \label{ent}

In the papers \cite{LuoLuo:2003} \cite{LuoZZhang:2004} Luo {\sl et
al.} proved that the covariance is not a good measure to quantify
entanglement properties of states; to this end, Wigner-Yanase
correlation was proposed.  Hereafter we recall the more general
definition of {\it metric adjusted correlation} (or {\it $f$-correlation}) introduced in \cite{Hansen:2006b}
\cite{GibiliscoImparatoIsola:2007}; Wigner-Yanase correlation is just
a particular example of metric adjusted correlation.  We show, by the same example
as in \cite{LuoZZhang:2004}, that the general metric adjusted correlation has the
same properties of Wigner-Yanase correlation with respect to entanglement.

\begin{definition}
    For $A,B \in M_{n,sa}$, $\rho \in {\cal D}_n^1$ and $f \in {\cal
    F}_{op}$, the {\sl metric adjusted correlation} (or {\sl
    $f$-correlation}) is defined as
    $$
    {\rm Corr}_{\rho}^{f}(A,B):= \frac{f(0)}{2} \langle i[\rho,A],
    i[\rho,B] \rangle_{\rho,f} = {\rm Tr}(\rho AB)-{\rm Tr}(m_{\tilde
    f}(L_{\rho},R_{\rho})(A) \cdot B).
    $$
\end{definition}

\begin{proposition} {\rm \cite{GibiliscoImparatoIsola:2007}} \label{gib}
    \[
    \begin{array}{rcl}
	{\rm Corr}_{\rho}^{f}(A,B)&=&{\rm Cov}_\rho(A,B)-{\rm
	Tr}(m_{\tilde f}(L_\rho,R_\rho)(A_0)B_0)\\[12pt]
	&=&\Cov_\rho(A,B)-\sum_{i,j}m_{\tilde
	f}(\lambda_i,\lambda_j)a_{ij}b_{ji}
    \end{array}
    \] 
\end{proposition}

\begin{proposition}  {\rm \cite{GibiliscoImparatoIsola:2007}} \label{pure}
    If $\rho$ is pure, then
    $$
    {\rm Corr}^f_{\rho}(A,B)= {\rm Cov}_{\rho}(A,B) \qquad \qquad
    \forall f \in {\cal F}_{op}.
    $$
\end{proposition}

Consider, now, 
\[ 
\rho:=\frac12\left(
\begin{array}{cccc}
    1 & 0 & 0 & 0 \\
    0 & 0 & 0 & 0 \\
    0 & 0 & 0 & 0 \\
    0 & 0 & 0 & 1 \\
\end{array}
\right),\quad
\rho':=\frac12\left(
\begin{array}{cccc}
    1 & 0 & 0 & 1 \\
    0 & 0 & 0 & 0 \\
    0 & 0 & 0 & 0 \\
    1 & 0 & 0 & 1 \\
\end{array}
\right)
\]

Note that the first state is a mixture of two disentangled states
while the second is a Bell state which is maximally entangled (see
\cite{LuoZZhang:2004} \cite{LuoLuo:2003}).

Set
\[
A:=\frac12\left(
\begin{array}{cccc}
    1 & 0 & 0 & 0 \\
    0 & 1 & 0 & 0 \\
    0 & 0 & -1 & 0 \\
    0 & 0 & 0 & -1 \\
\end{array}
\right),\quad
B:=\frac12\left(
\begin{array}{cccc}
    1 & 0 & 0 & 1 \\
    0 & -1 & 0 & 0 \\
    0 & 0 & 1 & 0 \\
    1 & 0 & 0 & -1 \\
\end{array}
\right).
\]                           
                           
In \cite{LuoZZhang:2004} \cite{LuoLuo:2003} it is shown that
\begin{equation} \label{lu}
    {\rm Cov}_{\rho}(A,B) ={\rm Cov}_{\rho'}(A,B)=1,
\end{equation} 
while 
\[
{\rm Corr}^{f_{WY}}_{\rho}(A,B)=0,\quad\quad {\rm
Corr}^{f_{WY}}_{\rho'}(A,B)=1.  
\]

Due to Proposition \ref{gib}, this result holds more generally for
any $f$-correlation.

\begin{proposition} 
    For any $f\in\mathcal F_{op}$ one gets
    \[
    {\rm Corr}^f_{\rho}(A,B)= 0,\quad\quad {\rm Corr}^f_{\rho'}(A,B)=1.
    \]
\end{proposition}
\begin{proof}
    Since $\rho'$ is a pure state, from Proposition \ref{pure} and due
    to (\ref{lu}) one has that ${\rm Corr}^f_{\rho'}(A,B)={\rm
    Cov}_{\rho'}(A,B)=1$.

    Consider, now, the state $\rho$ and let $\{e_1,e_2,e_3,e_4 \}$ be
    the canonical basis.  A direct computation shows that its
    eigenvalues are $\lambda_1=\lambda_4=\frac12$ and
    $\lambda_2=\lambda_3=0$, and the corresponding eigenvectors are
    $\{e_1,e_4\}$ and $\{e_2,e_3\}$, respectively.  Observe that $A$
    and $B$ are centered with respect to both the states $\rho$ and
    $\rho'$ (namely $A=A_0$, $B=B_0$).  Moreover, since the
    eigenvectors are the canonical basis one gets $A_{ij}=a_{ij}$ and
    $B_{ij}=b_{ij}$.

    This implies that
    \[
    \begin{array}{rcl}
	\displaystyle \sum_{i,j=1}^4m_{\tilde
	f}(\lambda_i,\lambda_j)a_{ij}b_{ji}&=&
	\displaystyle\sum_{i=1}^4m_{\tilde
	f}(\lambda_i,\lambda_i)A_{ii}B_{ii}\\[12pt]
	&=&\displaystyle\sum_{i=1}^4\lambda_iA_{ii}B_{ii}=1,
    \end{array}
    \]
    where we used the mean property that $m_f(x,x)=x$, for any non
    negative $x$.  Again from (\ref{lu}), one obtains ${\rm
    Corr}^f_{\rho}(A,B)= 0$.  
\end{proof}

\end{document}